\documentclass[fleqn,10pt]{SICE_FES26}

\usepackage{amsmath,latexsym} 
\usepackage{amssymb}  
\usepackage{amsfonts}  
\usepackage{cite,color}

\newtheorem{thm}{Theorem}

\newtheorem{lemma}[thm]{Lemma}

\newtheorem{prop}[thm]{Proposition}

\allowdisplaybreaks

\title{On the Oja-Flow-Based Low-Rank Approximation of Kalman-Bucy Filters for Linear Time-Varying Systems
}

\author{Kentaro Ohki${}^{1\dagger}$}
\speaker{Kentaro Ohki}

\affils{${}^{1}$Department of Applied Computer Engineering, Tokai University, Hiratsuka, Kanagawa, Japan\\
(E-mail: ohki@tokai.ac.jp)\\
}
\abstract{
This paper studies a low-rank Kalman-Bucy filtering framework for linear time-varying systems through the tracking analysis of Oja's principal component flow.
Under structured assumptions on the eigenspaces and their time variation, we show that the Oja flow can remain in a neighborhood of the time-varying dominant subspace by tuning a parameter of the flow, rather than tracking it exactly. 
These restrictive assumptions identify a tractable class of linear time-varying systems and provide a theoretical basis for low-rank filtering, as illustrated by a numerical experiment.
}


\begin{document}

\maketitle


\section{Introduction}

Kalman filter \cite{kalman1960nal} and Kalman-Bucy filter \cite{kalman1961nrl} are common estimation tools for discrete-time and continuous-time linear stochastic dynamical systems, respectively.  
They are also extended to nonlinear systems with some Gaussian approximation of the stochastic process such as linearization of the system or Laplace approximation of the noise statistics. 
These filters need to calculate a matrix Riccati equation, which is computational bottleneck if the system's dimension is huge.  
Therefore, we need to lighten the calculation of the Riccati equation. 

There are two ways to address this problem: model order reduction and low-rank approximation. 
The model order reduction is a tradition and wide variety of the methods have been proposed \cite{antoulas2005approximation,AntoulasBeattieGuegercin2020,breiten2021balancing}. 
Some of these methods preserve the system's properties, like stability and passivity.  
On the other hand, low-rank approximation for filter and controller synthesis is an alternative method that does not require the model order reduction, which does not govern approximation errors of the model. 
The continuous QR decomposition algorithm that is an iterative algorithm to obtain the dominant modes is an example of this manner, and applied to low-rank Kalman-Bucy filter for time-varying systems \cite{tranninger2022detectability}. 
The continuous QR decomposition algorithm is essentially the same as the Oja's principal component flow (Oja flow) \cite{TsuzukiOhki2025global}, and exponential convergence for time-invariant matrices is proved under a spectrum gap condition \cite{kuo2025asymptotic,TsuzukiOhki2025global}, which means that the Oja flow can track the dominant subspace of a class of time-varying systems. 

As described above, the convergence rate is exponential, meaning that practically the perfect dominant sub-eigenspace tracking is impossible for time-varying systems. 
Nevertheless, in \cite{tranninger2022detectability}, the authors only consider the perfect tracking case and show the performance of the low-rank Kalman-Bucy filter.  
This paper examines the online Oja-flow for time-varying matrices and then applies the result to low-rank Kalman-Bucy filter. 
Because of the difficulty of the analysis of general LTV systems, we only focus on LTV systems with a smooth, bounded coefficient matrices.  

The contribution of this paper is twofold. First, we analyze the tracking behavior of the time-varying Oja flow and show that, under structured assumptions, the solution remains in a neighborhood of the time-varying dominant subspace by tuning a parameter of the flow. Second, we use this result to interpret and design a low-rank Kalman-Bucy filter for LTV systems. Thus, the theoretical core of the paper is the dominant-subspace tracking analysis, while the filtering result is its application.

The rest of the paper is organized as follows. 
We briefly introduce the Oja flow, which is the key tool for low-rank Kalman-Bucy filters, and its convergence property for constant square matrices in the next section. 
Section 3 provides the main contributions of this paper. First we analyze the tracking performance of a two-dimensional rotating matrix and then generalize the results for a wider class of time-varying matrices. 
Section 4 demonstrates a numerical example and discuss the results. 
Section 5 concludes the paper and comment on the future direction of this work.

\paragraph*{\textbf{Notation:} }
The sets of real and complex numbers are ${\mathbb{R}}$ and $\mathbb{C}$. The set of $n \times m$ real matrices is ${\mathbb{R}}^{n\times m}$. $I_n$ is the $n \times n$ identity matrix, and $0_{n,m}$ is the $n \times m$ zero matrix. For simplicity, we denote $0_{n} = 0_{n,n}$; if the dimension is trivial, $0$ is also used. 
$A^{\top}$ and $A^{\dagger}$ denote the transpose and Hermitian conjugate of a matrix $A$. $\|x\|$ is the Euclidean norm for a vector $x$ and $\| A \| _{\rm ind}$ is its induced norm for a matrix $A$. For a symmetric matrix $A$, $A > 0$ ($A \geq 0$) indicates it is positive-definite (semidefinite). $A^{1/2}$ denotes the unique positive-semidefinite square root. The eigenvalues of a square matrix $A \in {\mathbb{C}}^{n \times n}$ are ordered such that ${\mathrm{Re}}(\lambda_1(A)) \geq \dots \geq {\mathrm{Re}}(\lambda_n(A))$. The corresponding (generalized) eigenvector is $\psi_i(A)$, and the matrix of eigenvectors is $\Psi(A) := [\psi _{1}(A) , \dots , \psi _{n}(A)]$. The Stiefel manifold is ${\mathrm{St}}(r,n) := \{X \in {\mathbb{R}}^{n \times r} \mid X^{\top}X = I_r\}$. 
For a given $U \in {\mathbb{R}}^{n\times r}$ and $\delta >0$, 
neighborhood is denoted by ${\mathcal{N}}_{\delta} (U) := \{ U+X \ | \ X \in {\mathbb{R}}^{n\times r},\ \| X \| _{\rm ind} < \delta \}$.

\section{Previous Works}

Yamada and Ohki proposed a low--rank Kalman-Bucy filter \cite{yamada2020lowrank}, which was analyzed by  \cite{yamada2021comparison,TsuzukiOhki2024}. 
The key idea is the use of a principal component flow proposed in \cite{bonnabel2012geometry}. 
	The Oja's principal component flow (Oja flow for short) is the following matrix differential equation for $A \in {\mathbb{R}}^{n\times n}$:  
    \begin{align}
    \varepsilon \frac{d}{dt} U(t) = & ( I_{n} - U(t) U(t) ^{\top} ) A U(t) .
    \label{eq:Oja_flow}
    \end{align}
    where $\varepsilon >0$. 
    If $U(0) \in {\mathrm{St}}(r,n)$ and the spectrum gap condition ${\mathrm{Re}}(\lambda _{r}(A) - \lambda _{r+1}(A)) >0$ holds, then $U(t) \in  {\mathrm{St}}(r,n)$ for all $t\geq 0$ even for $A \in {\mathbb{R}}^{n\times n}$ is not symmetric.  
    Furthermore, for almost all initial matrix $U(0) \in {\mathrm{St}}(r,n)$, $U(t)$ converges an element of the following invariant subset 
    \begin{align*}
    &{\mathcal{U}}_{r} (A) 
    \\
    := &\left\{ 
    \Psi (A) \begin{bmatrix} K_{r} \\ 0 _{n-r,r}\end{bmatrix} \in {\mathrm{St}}(r,n) \ | \  K_{r} \in {\mathbb{C}}^{r\times r}
    \right\}
    .
    \end{align*}
    For the details of the analysis of the Oja flow, see \cite{TsuzukiOhki2025global}.

\section{Main Results}

 From Theorem 12 of \cite{TsuzukiOhki2025global}, the convergence rate of the Oja flow \eqref{eq:Oja_flow} with a constant matrix is proportional to the difference of real part of eigenvalues ${\mathrm{Re}} (\lambda _{r}(A) - \lambda _{r+1}(A))$ and accerelated by the tuning parameter $\varepsilon \in (0,1]$. 
    Here, natural questions arise: Can the Oja flow track the dominant mode if $A$ is time-varying? How should we adjust $\varepsilon \in (0,1]$ in \eqref{eq:Oja_flow}? 
    This paper cannot provide comprehensive answers to the questions, as the analysis of time-varying matrices is significantly more challenging.   
    Here, we consider only bounded time-varying matrices and attempt to answer the questions partially.  
    Consider the following time-varying Oja flow
    \begin{align}
    \varepsilon \frac{d}{dt} U(t) = & ( I_{n} - U(t) U(t) ^{\top} ) A(t) U(t) .
    \label{eq:time-varying_Oja_flow}
    \end{align}
    with $U(t_{0}) \in {\mathrm{St}}(r,n)$ at an initial time $t_{0} \in {\mathbb{R}}$, where $A: {\mathbb{R}} \to {\mathbb{R}}^{n\times n}$ is bounded; $\sup _{t \geq t_{0}} \| A(t) \| _{\rm ind} <\infty$.  
    Notice that the equilibrium points of \eqref{eq:time-varying_Oja_flow} may not exist for the time-varying case. 
    We denote ${\mathcal{U}}_{r}(A(t))$ as a set of zero solutions of the right-hand side of Eq. \eqref{eq:time-varying_Oja_flow}. 
    	The following result is a generalization of Theorem 9 in \cite{TsuzukiOhki2025global}.

	\begin{prop}\label{prop:time-varying_positive_part_convergence}
	For a given matrix-valued bounded function $A: {\mathbb{R}} \to {\mathbb{R}}^{n\times n}$ and arbitrarily positive constant $c>0$, take $a>0$ such that $A_{\rm sym}(t) + aI_{n} >cI_{n}$ for all $t\in {\mathbb{R}}$. 
       If $U(t_{0}) = U_{0} \in {\mathbb{R}}^{n\times r}$ is a full rank matrix, then $U(t)$ converges to ${\mathrm{St}}(r,n)$ exponentially with the convergence rate $-4\alpha c/\varepsilon$, where $\alpha$ is defined in Lemma \ref{lem:well_definedness_time_varying_Oja_flow} in Appendix. 
	\end{prop}	
	
	The proof is almost the same as that of Theorem 9 in \cite{TsuzukiOhki2025global}, and is given in Appendix \ref{proof:time-varying_positive_part_convergence}. 
    From Proposition \ref{prop:time-varying_positive_part_convergence}, the Oja flow for bounded time-varying matrices can remain on ${\mathrm{St}}(r,n)$ by choosing a large $a>0$.

	Although the Oja flow \eqref{eq:time-varying_Oja_flow} can also converge to the Stiefel manifold, it is unclear if the solution $U(t)$ of \eqref{eq:time-varying_Oja_flow} can also track the dominant eigen subspaces when the dominant eigen subspaces vary. 
	In the general tracking problem, perfect tracking is impossible if the target's behavior is unknown. 
	The following is a toy problem. 
	
    \begin{prop}\label{prop:rotating_A}
    Let $(r,n) = (1,2)$ and consider the Oja flow \eqref{eq:Oja_flow} with the following time-varying matrix. 
    \begin{align*}
    A (t) = &
    R(t)
    \Lambda 
    R(t)^{\top},
    \\
    R(t) :=& \begin{bmatrix}
    \cos (\omega t) & \sin (\omega t) \\ -\sin (\omega t) & \cos (\omega t)
    \end{bmatrix},
     \ 
     \Lambda 
   := \begin{bmatrix}
    \lambda _{1} & 0 \\ 0 & \lambda _{2}, 
    \end{bmatrix}, 
    \end{align*}
    where $\lambda _{1} > \lambda _{2}$ and $\omega >0$. 
    Then, for any $\bar{\delta} \in (0,1)$, there exists $\varepsilon \in (0,1]$ such that $U(t)$ remains in the neighborhood of the dominant eigenvector $\psi _{1}(A(t)) =  \begin{bmatrix} \cos (\omega t) & -\sin (\omega t) \end{bmatrix}^{\top}$
    \begin{align*}
    & {\mathcal{N}}_{\bar{\delta}}( \psi _{1}(A(t))) 
    \\
    := & \big{\{} X \in {\mathrm{St}}(1,2) \ \big{|} \ 
    X = \psi _{1}(A(t))  + Y, 
    \\ 
    & \hspace{3cm} Y \in {\mathbb{R}}^{2}, \| Y \| _{\rm ind}< \bar{\delta}
    \big{\}}
    \end{align*}
    for $t\geq t_{0} \geq 0$ if $U(t_{0}) \in {\mathcal{N}}_{\delta }( \psi _{1}(A(t_{0})))$ where $\delta \in (0, \bar{\delta}/2)$. 
    \end{prop}

    \begin{proof}
    To move on the rotating frame, we consider the differential equation of $K(t) := R(t) ^{\top} U(t)$, 
    \begin{align*}
    &\varepsilon \frac{d}{dt} K(t) 
    \\
    =&
    \varepsilon \omega 
    \begin{bmatrix}
    0 & - \cos (2\omega t) \\ \cos (2\omega t) & 0
    \end{bmatrix} K(t)
    \\
    &
    +
    (I_{2} - K(t) K(t)^{\top}) \Lambda K(t)
    .
    \end{align*}
    Hence, the first term can be regarded as a perturbation or small modeling error. 
    Notice that if $\omega =0 $, then ${\mathcal{U}}_{1} := {\mathcal{U}}_{1}({\mathrm{diag}}(\lambda_{1}, \lambda _{2})) = \{ \begin{bmatrix} \pm 1 & 0 \end{bmatrix} ^{\top} \}$ is the stable equilibrium set. 
    However, ${\mathcal{U}}_{1}$ is no longer the equilibrium set when $R(t)$ is time-varying. 
    Since the second term of the right-hand side of the above equation forces the solution to ${\mathcal{U}}_{1}$, consider the linearization around an element of ${\mathcal{U}}_{1}$. 
    Let $K(t) = \bar{K} + K_{\delta}(t)$, where $\bar{K} \in {\mathcal{U}}_{1}$ and $K_{\delta}(t_{0}) \in {\mathcal{B}}_{\delta} := \{ X \in {\mathbb{R}}^{2\times 1} \ | \ \| X \| _{\rm ind} < \delta \ \}$ for a small positive $\delta \in (0, \bar{\delta} /2)$ at some $t_{0}\geq 0$. 
    Then, for $t\geq t_{0}$, 
    \begin{align*}
    & \varepsilon \frac{d}{dt} K_{\delta}(t) 
    \\
    =&
    -\varepsilon \omega \begin{bmatrix}
    0 & - \cos (2\omega t) \\ \cos (2\omega t) & 0
    \end{bmatrix} ( \bar{K} + K_{\delta}(t) )
    \\
    & +
    ( I_{2} - \bar{K} \bar{K}^{\top} ) \Lambda K_{\delta}(t)
    \\
    &
    -
    (\bar{K} K_{\delta} (t)^{\top} + K_{\delta}(t) \bar{K}^{\top} ) \Lambda \bar{K}
    + O (\delta ^2) 
    \end{align*}
    and without loss of generality, take $\bar{K} = \begin{bmatrix} 1 & 0 \end{bmatrix} ^{\top}$.  
    Then, the element-wise equation is 
    \begin{align*}
    \varepsilon \frac{d}{dt} K_{\delta , 2}(t) =&
     \varepsilon \omega (1+K_{\delta , 1}(t))
     \cos (2\omega t)
     \\ & 
    -(\lambda _{1} - \lambda _{2})
     K_{\delta,2}(t)
    + O (\delta ^2) 
    \end{align*}
    and its solution is
    \begin{align*}
    K_{\delta,2} (t) = & \exp \left( -\frac{\lambda _{1} - \lambda _{2}}{\varepsilon}(t-t_{0}) \right) K_{\delta,2}(t_{0})
    \\ &
    +
    \omega \int _{t_{0}}^{t} \exp \left( -\frac{\lambda _{1} - \lambda _{2}}{\varepsilon} (t- \tau ) \right) 
    \\ & \hspace{1cm}
    \times (1+ K_{\delta ,1}(\tau )) \cos (2\omega \tau ) d\tau 
    \\ & 
    + O (\delta ^2  ) 
    .
    \end{align*}
    Therefore, 
    \begin{align*}
    | K_{\delta,2} (t) | \leq &  \exp \left( -\frac{\lambda _{1} - \lambda _{2}}{\varepsilon}(t-t_{0}) \right) |K_{\delta,2}(t_{0}) |
    \\ & 
    +
    \omega (1 + c(t) ) \int _{t_{0}}^{t} \exp \left( -\frac{\lambda _{1} - \lambda _{2}}{\varepsilon}\tau \right) d\tau 
    \\ & + O (\delta ^2  ) 
    \\
    \leq &
    \exp \left( -\frac{\lambda _{1} - \lambda _{2}}{\varepsilon}(t-t_{0}) \right) \delta 
    \\ & 
    +
    \frac{(1 + c(t) )\varepsilon \omega }{\lambda _{1} - \lambda _{2}} 
    \Bigg{(} 
    \exp \left( -\frac{\lambda _{1} - \lambda _{2}}{\varepsilon}t_{0} \right)
    \\ & \hspace{2cm}
    -
    \exp \left( -\frac{\lambda _{1} - \lambda _{2}}{\varepsilon}t \right)
    \Bigg{)}
    \\ & 
    + O (\delta ^2  ) 
    \\
    \leq &
    \exp \Bigg{(} -\frac{\lambda _{1} 
    - \lambda _{2}}{\varepsilon}(t-t_{0}) \Bigg{)} \delta 
    \\ &
    +
    \frac{(1+c(t))\varepsilon \omega }{\lambda _{1} - \lambda _{2}} 
    + O (\delta ^2  ) ,
    \end{align*}
    where $c(t) := \max _{\tau \in [t_{0},t]} |K_{\delta ,1}(\tau)|$. 
    Notice that $c_{\infty} := \sup _{t\geq t_{0}} c(t) < \infty$ because $\Psi(A(t))=R(t)$ is bounded and $U(t) \in {\mathrm{St}}(1,2)$. 
    Therefore, if for any $\delta >0$, taking $\varepsilon  \in (0,1]$ such that $\varepsilon  < \min \{ \delta , \delta (\lambda _{1} - \lambda _{2})/\omega (1+c_{\infty}) \}$ ensures the solution $ K_{\delta}(t) \in {\mathcal{B}}_{\bar{\delta}}$ for $t\geq t_{0}$. 
    The dominant mode is then $R(t) \bar{K} = \pm \begin{bmatrix} \cos (\omega t) & -\sin (\omega t) \end{bmatrix} ^{\top}$ and the solution remains near the dominant mode. 
    \end{proof}

    The proof of Proposition \ref{prop:rotating_A} gives an insight into how to choose $\varepsilon$ for tracking dominant modes. 
    To track the dominant modes by time-varying Oja flow \eqref{eq:time-varying_Oja_flow}, the acceptance neighborhood ${\mathcal{N}}_{\delta} (\bar{U}(t))$ of $\bar{U}(t) \in {\mathcal{U}}_{r}(A(t))$ for $\delta >0$, $\| dA(t) /dt \| _{\rm ind}$, and the convergence rate ${\mathrm{Re}}(\lambda _{r}(A(t)) - \lambda _{r+1}(A(t)))$ are keys to determine $\varepsilon $.

    \begin{thm}\label{thm:time_varying_tracking_performance}
    Let $A(t) = \Psi (t) \Lambda (t) \Psi (t) ^{-1} $, where $\Psi (t) :=[ \psi _{1} (A(t)),\dots , \psi _{n}(A(t)) ]$, $\Lambda (t) := \mbox{block-diag}[\Lambda _{r}(t) , \Lambda _{\perp}(t)]$, with Jordan form matrices $\Lambda _{r}(t) \in {\mathbb{C}}^{r \times r}$, and $\Lambda _{\perp} (t) \in {\mathbb{C}}^{(n-r) \times (n-r)}$. 
    Assume the following conditions hold. 
    \begin{enumerate}
    \item $\Lambda (t)$ is diagonal for $t\geq t_{0}$. 
    
    \item $\Lambda (t)$ is continuous in $t\geq t_{0}$ and ${\mathrm{Re}}(\lambda _{r} (\Lambda _{r}(t)) - \lambda _{1}(\Lambda _{\perp} (t))) \geq c_{1}$ for any $t \geq t_{0}$ for some $c_{1} >0$. 
    
    \item $\Psi (t)$ and its inverse $\Psi (t)^{-1}$ and their derivatives are continuous and bounded.  
    
    \item $\Psi (t)^{\dagger} \Psi (t)$ is a constant matrix. 
    
    \item $\Psi (t)^{-1} {\mathcal{U}}_{r}(A(t)) $ is time invariant. 
    
    \end{enumerate}
    Then, there exists $\varepsilon >0$ such that 
    $U(t)$ starting from $U(t_{0}) \in {\mathcal{N}}_{\delta} (\bar{U}(t_{0}))$, where $\bar{U}(t_{0}) \in {\mathcal{U}}_{r}(A(t_{0}))$ and $\delta >0$,  remains in ${\mathcal{N}}_{\delta} (\bar{U}(t))$. 
    \end{thm}
    
    The fourth assumption is a generalization of the case of a rotation matrix; it allows that $\Psi (t)$ consists of a multiplication between a time-varying rotation matrix $R(t)$ and a non-singular constant matrix $\Psi ^{\prime}$, i.e., $\Psi (t) = R(t) \Psi ^{\prime}$. 
    
    \begin{proof}
    Let $K (t) = \Psi (t)^{-1} U(t)$ and consider its time evolution. 
	\begin{align*}
	\varepsilon \frac{d}{dt} K(t) 
	=& - \varepsilon \Psi (t) ^{-1} \frac{d}{dt}\Psi (t) K(t) 
	\\ & 
	+ (I_{n} - K(t)K(t)^{\dagger} \Psi (t)^{\dagger} \Psi (t)  ) \Lambda (t) K(t).
	\end{align*}
    Notice that since $\Psi (t)^{-1}$ is bounded and $U(t) \in {\mathrm{St}}(r,n)$, $K(t)$ is bounded too. 
	Next, we consider the perturbed trajectory from $\Psi (t) ^{-1} {\mathcal{U}}_{r}(A(t))$. 
	Notice that ${\mathcal{U}}_{r}(A(t))$ is in a time varying subspace of ${\mathbb{R}}^{n\times r}$, while $\Psi (t) ^{-1} {\mathcal{U}}_{r}(A(t))$ is in a {\em fixed} subspace in ${\mathbb{C}}^{n\times r}$. 
	This means that with the assumption 5, $\bar{K} \in \Psi (t) ^{-1} {\mathcal{U}}_{r}(A(t))$ is of the form $\bar{K} = \begin{bmatrix}  \bar{K}_{r}^{\dagger}  & 0_{n-r,r}^{\dagger}  \end{bmatrix} ^{\dagger}$ with a non-singular matrix $\bar{K}_{r} \in {\mathbb{C}}^{r\times r}$ for $t\geq 0$. 
	From the assumption 4, 
	\begin{align*}
	\frac{d}{dt} ( \bar{K}^{\dagger} \Psi (t) ^{\dagger} \Psi (t) \bar{K})
	=
	0_{r,r}
	.
	\end{align*}
	From the assumption 3, there exists 
	\begin{align*}c_{2} := \sup _{t\geq t_{0}} \|  \Psi (t) ^{-1}  \frac{d}{dt}\Psi (t) \|  _{\rm ind}. 
	\end{align*}
	Consider $K(t) = \bar{K} + L (t) $ where $\| L (t_{0}) \| _{\rm ind} < \delta ^{\prime} := \delta /c_{2}$. 
	Then, 
	\begin{align*}
	& \varepsilon \frac{d}{dt} (\bar{K} + L(t) ) 
	\\
	= &  -\varepsilon \Psi (t) ^{-1} \frac{d\Psi (t)}{dt} ( \bar{K} +L(t)  ) 
	\\ & 
	- L(t)\bar{K}^{\dagger} \Psi (t)^{\dagger} \Psi (t)   \Lambda (t) \bar{K}
	\\ &
	- \bar{K} L(t)^{\dagger} \Psi (t)^{\dagger} \Psi (t)  \Lambda (t) \bar{K} 
	\\ & 
	+ ( I_{n} - \bar{K} \bar{K} ^{\dagger} \Psi (t)^{\dagger} \Psi (t)  ) \Lambda (t) L(t) 
	\\
	&
	+ O(\| L(t) \| _{\rm ind}^{2}). 
	\end{align*}
	Ignoring the higher terms, the linearized equation becomes 
	\begin{align*}
	\varepsilon \frac{d}{dt} L (t)  
	=&
	- \varepsilon \Psi (t) ^{-1} \frac{d\Psi (t)}{dt} (\bar{K} + L (t) )
	\\
	& + (I_{n} - \bar{K} \bar{K}^{\dagger} \Psi (t) ^{\dagger} \Psi (t) ) \Lambda (t) L(t) 
	\\
	& + 
	(\bar{K}  L(t)^{\dagger} + L (t) \bar{K}^{\dagger} )
	\\ & \times  \Psi (t)^{\dagger} \Psi (t) \Lambda (t) \bar{K}
	\end{align*}
	For $L(t) = \begin{bmatrix} L_{r}(t)^{\dagger}  & L_{\perp}(t)^{\dagger}  \end{bmatrix} ^{\dagger}$, if $L_{\perp} (t) \in {\mathbb{C}}^{(n-r)\times r}$ becomes zero, then $U(t)$ is in ${\mathcal{U}}_{r}(A(t))$. 
	In the rest of the proof, we establish how we can attenuate $L_{\perp}(t)$ by adjusting $\varepsilon \in (0,1]$. 

	Let $L_{\perp}^{\prime}(t) := L_{\perp} (t) \bar{K}_{r}^{-1}$.  
	Similarly with the proof of Theorem 1 of \cite{TsuzukiOhki2024}, the time evolution equation of $L_{\perp}^{\prime} (t) $ is then 
	\begin{align*}
	\varepsilon \frac{d}{dt} L_{\perp}^{\prime} (t) =& 
	-C_{1}(t) ( \bar{K}_{r} +  L_{r}^{\prime}(t))
	-
	\varepsilon C_{2} (t) L_{\perp}^{\prime}(t)
	\\ & 
	+
	\Lambda _{\perp} (t) L_{\perp}^{\prime} (t) 
	-
	 L_{\perp}^{\prime} (t) \Lambda _{r}(t) , 
	\end{align*}
	where 
	\begin{align*}
	C_{1}(t)  :=& \begin{bmatrix} 0_{n-r,n-r} \\ I_{n-r} \end{bmatrix} ^{\top} 
	 \Psi (t) ^{-1} \frac{d}{dt}\Psi (t)\begin{bmatrix} I_{r} \\ 0_{n-r,r} \end{bmatrix}, 
	 \\
	 C_{2}(t)  := &\begin{bmatrix} 0_{n-r,n-r} \\ I_{n-r} \end{bmatrix} ^{\top}
	 \Psi (t) ^{-1} \frac{d}{dt}\Psi (t)
	 \begin{bmatrix} 0_{r,n-r} \\ I_{n-r} \end{bmatrix}, 
	 \end{align*}
	 $D(t) := \bar{K}_{r}(t) \frac{d}{dt}\bar{K}_{r}(t)^{-1} $ and 
	$\bar{U}(t) \in {\mathcal{U}}_{r}(A(t))$. 
	Then, 
	\begin{align*}
	L_{\perp}^{\prime} (t) = & 
	\Phi _{C_{2} + \Lambda _{\perp}/ \varepsilon }(t) L_{\perp}^{\prime}(t_{0}) \Phi _{-D -\Lambda _{r} /\varepsilon }(t) 
	\\ & 
	- 
	\int _{t_{0}}^{t} \Phi _{C_{2} + \Lambda _{\perp}/ \varepsilon }(t-\tau) 
	(\bar{K}_{r} + C_{1}(\tau )) 
	\\ & \hspace{2cm}
	\times \Phi _{-D -\Lambda _{r} /\varepsilon }(t-\tau )  d\tau . 
	\end{align*}
	Note that $C_{1}(\bullet )$, $C_{2}(\bullet )$, and $D(\bullet )$ are bounded by $c_{2}$. 
	Since the maximal singular value of a square matrix is larger than or equal to the maximum eigenvalue of the symmetrized matrix, 
	\begin{align*}
	&\frac{d}{dt}  ( \Phi _{C_{2} + \Lambda _{\perp}/ \varepsilon }(t) ^{\dagger} \Phi _{C_{2} + \Lambda _{\perp}/ \varepsilon }(t) )
	\\
	=& \Phi _{C_{2} + \Lambda _{\perp}/ \varepsilon }(t) ^{\dagger} 
	(C_{2}(t)+C_{2}(t) ^{\dagger} + 2 \Lambda _{\perp}(t)/ \varepsilon )
	 \\
	 & \quad \times \Phi _{C_{2} + \Lambda _{\perp}/ \varepsilon }(t)
	 \\
	 \leq & 
	 2 \Phi _{C_{2} + \Lambda _{\perp}/ \varepsilon }(t) ^{\dagger} 
	(c_{2} I_{n-r} +  \Lambda _{\perp}(t)/ \varepsilon )
	 \Phi _{C_{2} + \Lambda _{\perp}/ \varepsilon }(t)
	 \\
	 \leq &
	 2 \Phi _{c_{2}I_{n-r} + \Lambda _{\perp}/ \varepsilon }(t) ^{\dagger} 
	(c_{2} I_{n-r} +  \Lambda _{\perp}(t)/ \varepsilon )
	\\ & \quad \times 
	 \Phi _{c_{2}I_{n-r} + \Lambda _{\perp}/ \varepsilon }(t)
	,
	\end{align*}
	where $\Phi _{c_{2}I_{n-r} + \Lambda _{\perp}(t)/ \varepsilon }(t)$ has the following explicit representation; 
	\begin{align*}
	& \Phi _{c_{2}I_{n-r} + \Lambda _{\perp}/ \varepsilon }(t)
	\\
	=&
	\exp \left(  c_{2}I_{n-r}(t-t_{0}) + \int _{t_{0}}^{t} \Lambda _{\perp}(\tau ) d\tau / \varepsilon  \right)
	.
	\end{align*}
	Notice that $\Phi _{c_{2}I_{n-r} + \Lambda _{\perp}/ \varepsilon }(t)$ is a diagonal matrix. 
	Therefore, the following holds;
	\begin{align*}
	\| \Phi _{C_{2} + \Lambda _{\perp}/ \varepsilon }(t) \| _{\rm ind} 
	= 
	{\mathrm{e}}^{c_{2}(t-t_{0}) + \int _{t_{0}}^{t} {\mathrm{Re}}(\lambda _{r+1}(\tau )) d\tau  / \varepsilon) }.
	\end{align*}
	Similarly, we have 
	\begin{align*}
	\| \Phi _{-D -\Lambda _{r} /\varepsilon }(t) \| _{\rm ind} 
	=
	{\mathrm{e}}^{c_{2}(t-t_{0}) -\int _{t_{0}}^{t} {\mathrm{Re}} (\lambda _{r}(\tau )) d\tau  / \varepsilon) }. 
	\end{align*}
	Taking the induced norm for $L_{\perp}^{\prime}(t)$ and using the arguments of Proposition \ref{prop:rotating_A} gives 
	\begin{align*}
	& \| L_{\perp}^{\prime} (t) \| _{\rm ind}
	\\
	\leq & 
	\| L_{\perp}^{\prime}(t_{0}) \| _{\rm ind} 
	 {\mathrm{e}}^{2c_{2}(t-t_{0}) + \int _{t_{0}}^{t} {\mathrm{Re}}( \lambda _{r+1}(\tau ) - \lambda _{r}(\tau ) )d\tau  / \varepsilon) }
	\\ &
	+ ( \| \bar{K}_{r} \| _{\rm ind} + c(t) ) 
	\\ & \times 
	 \int _{t_{0}}^{t} {\mathrm{e}}^{2c_{2}(\tau -t_{0}) + \int _{t_{0}}^{\tau } {\mathrm{Re}}( \lambda _{r+1}(s ) - \lambda _{r}(s ) )ds  / \varepsilon) }
	d\tau
	\\
	\leq &
	 {\mathrm{e}}^{(2c_{2} - c_{1} / \varepsilon )(t-t_{0})  } \delta 
	\\ &
	+ ( \| \bar{K}_{r} \| _{\rm ind} + c(t) ) 
	 \int _{0}^{t - t_{0}} {\mathrm{e}}^{( 2c_{2} - c_{1}/\varepsilon ) \tau   } d\tau 
	 \\
	 =&
	 {\mathrm{e}}^{(2c_{2} - c_{1} / \varepsilon )(t-t_{0})  } \delta 
	\\ &
	+ \frac{\| \bar{K}_{r} \| _{\rm ind} + c(t)  }{c_{1}/\varepsilon - 2c_{2}}
	(1 - {\mathrm{e}}^{(2c_{2} - c_{1} / \varepsilon )(t-t_{0})  })
	\end{align*}
	where $c(t) := \max _{\tau \in [t_{0} ,t ]} \| C_{1} (\tau ) \|_{\rm ind} \leq c_2$.  
	If we take $\varepsilon < \min \{ c_{1}/2c_{2} , \ \delta c_{1} / ( \| \bar{K}_{r} \| _{\rm ind} + (1+2c_{2}) \delta ) \} $, $\| L_{\perp}^{\prime} \| _{\rm ind} < 2\delta < \bar{\delta}$. 
	\end{proof}
	
	Therefore, taking a small $\varepsilon >0$ enables the time-varying Oja flow \eqref{eq:time-varying_Oja_flow} to track the dominant subspaces.   
	
	The assumptions in Theorem \ref{thm:time_varying_tracking_performance} are restrictive and do not cover arbitrary LTV systems. Rather, they characterize a tractable class of smoothly varying systems for which dominant-subspace tracking can be analyzed explicitly. Extending the result beyond this class is left for future work.
	
	A numerical example to track a dominant subspace by the time-varying Oja flow \eqref{eq:time-varying_Oja_flow} with $A(t) = R(t) A_{0} R(t)^{\top} \in {\mathbb{R}}^{3\times 3}$, where
	\begin{align*}
	A_{0} :=& \begin{bmatrix}
	1 & 1 & 2 \\ 0 & 0 & 1 \\ 0 & 0 & -1
	\end{bmatrix}
	,\\ 
	R(t) :=&
	\begin{bmatrix}
	\cos (t) & -\sin(t) & 0 \\ \sin (t) & \cos(t) & 0 \\ 0 & 0 & 1
	\end{bmatrix}
	\\ & \times 
	\begin{bmatrix}
	1 & 0 & 0
	\\
	0 & \cos (t) & -\sin(t) \\ 0 & \sin (t) & \cos(t) 
	\end{bmatrix},
	\end{align*}
	 is shown in Figure \ref{fig_example_time_varying_semilog} in logarithmic scale. 
	 Notice that the matrix $A(t)$ satisfies the conditions of Theorem \ref{thm:time_varying_tracking_performance}.  
	Since $R(t)$ is a rotating matrix, ${\mathcal{U}}_{1}(A(t)) = \{ \pm R(t) \psi _{1} \}$, where $\psi _{1} = \psi _{1}(A_{0})$.  
	Each line shows a distance from the dominant subspace by using $\varepsilon = 1,0.1, 0.01$, respectively. 
	The results show that a small $\varepsilon$ makes $U(t)$ near the dominant subspace.

    	\begin{figure}[!htbp]
    	\centering
    	\includegraphics[keepaspectratio,width=\linewidth]{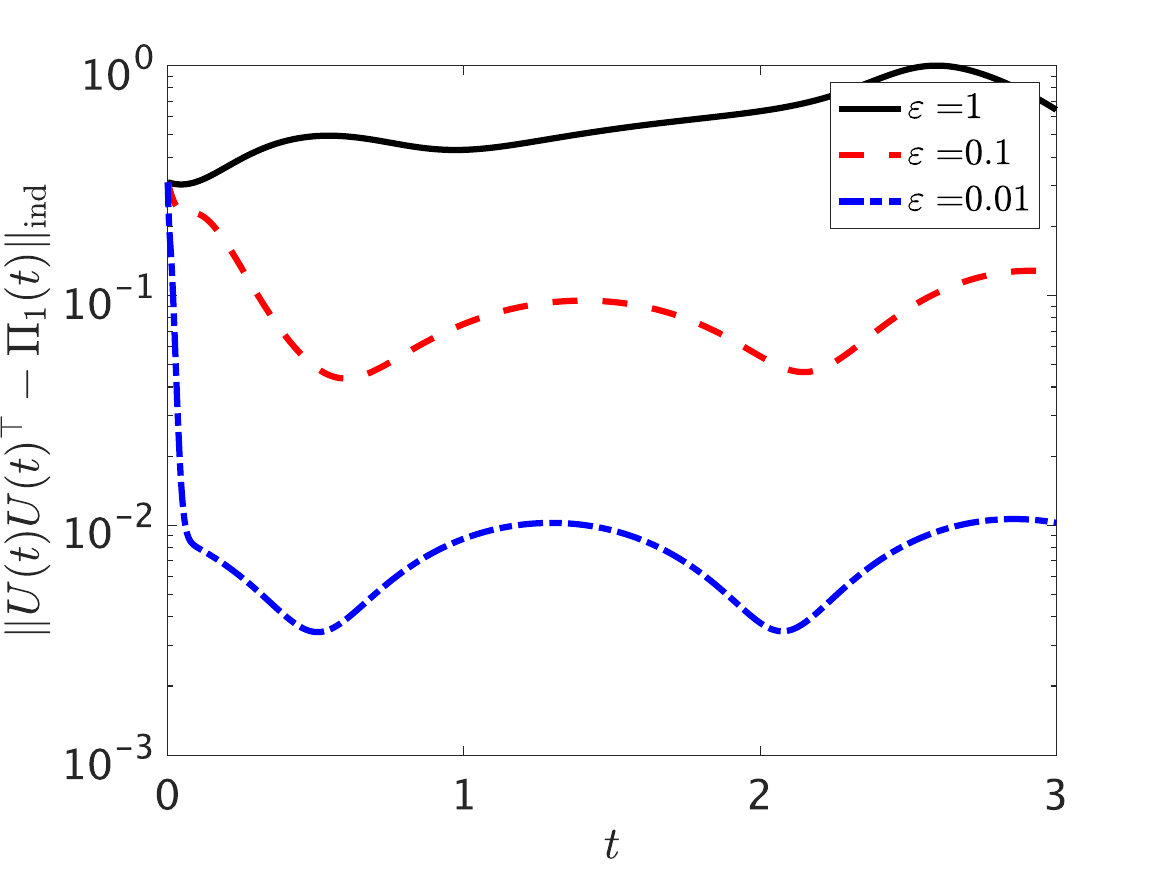}
    	\caption{Plot of $\| U(t) U(t)^{\top} - R(t) \psi _{1} \psi _{1}^{\top}R(t)^{\top} \| _{\rm ind}$ for each time $t$ by Forward Euler method with $A(t)+2I_{3}$ under different convergence parameter $ \varepsilon = 1, 0.1, 0.01$ with $U(0) = (\psi _{1} + \psi _{2} + \psi _{3})/\| \psi _{1} + \psi _{2} + \psi _{3} \|$. }
	\label{fig_example_time_varying_semilog}
	\end{figure}

\section{Numerical experiments}

In order to focus on time-varying systems, we employ an LTV system.  
\begin{align}
dx(t) =& A(t) x(t) dt + Gdw(t), \label{eq:state} 
\\ 
dy(t) =& Cx(t) dt + dv(t), \label{eq:observation}
\end{align}
where $x(t) \in {\mathbb{R}}^{n}$ is the state vector, $y(t) \in \mathbb{R}$ is the measurement output, $w(t) \in \mathbb{R}$ and $v(t) \in \mathbb{R}$ are mutually independent standard Wiener processes. The associated coefficients are the matrices with adequate dimensions. We utilize the following matrices:
\begin{align*}
A (t) =& R(t) A R(t)^{\top}, A = \begin{bmatrix}
0 & 1 & 0 & \dots & 0 \\
\vdots & \ddots & \ddots & & \vdots \\
0 & \dots  & \dots & 0 & 1 \\
1 & \dots &  & 1 & 1
\end{bmatrix}
,\\
B =& \begin{bmatrix} 0 &\dots &0 & 1\end{bmatrix}^{\top} , 
\
C = \begin{bmatrix} 1 & 0 &\dots &0\end{bmatrix},
\end{align*}
and $R(t) \in {\mathbb{R}}^{n\times n}$ is a rotation matrix defined as $R(t) = \exp ( (X - X^{\top}) t/2)$ where $X \in {\mathbb{R}}^{n\times n}$ is a random matrix whose antisymmetric part's maximum singular value is normalized to $10^{-4}$.  
Notice that $A$ is an unstable matrix in Hurwitz sense. 
To the best of the authors' knowledge, there is no solid method to reduce linear time-varying unstable systems. 
Assume that $x(0)$ follows a Gaussian distribution with mean $\mu$ and covariance matrix $P(0)$. 
Then the Kalman-Bucy filter of the system \eqref{eq:state} and \eqref{eq:observation} is 
\begin{align}
d\hat{x} (t) = & A(t) \hat{x}(t) dt + P(t) C^{\top} (dy(t) - C\hat{x}(t) dt),
\\
\frac{d}{dt}P(t) =& A(t) P(t) + P(t) A(t) + GG^{\top} \nonumber 
\\ & \quad - P(t) C^{\top}C P(t),
\end{align}
where $\hat{x}(t) \in {\mathbb{R}}^{n}$ is the conditional expectation of $x(t)$ on the measurement record up to $t\geq 0$, and $P(t) \in {\mathbb{R}}^{n\times n}$ is its conditional covariance matrix. 

We take the $r$ as the number of the unstable eigenvalues of $A$.  
Let $ U_{\varepsilon}(t)$ be the solution of \eqref{eq:time-varying_Oja_flow} and $ V_{\varepsilon}(t)$ be the solution of the time-varying Oja flow for $A(t)^{\top}$ with the initial condition $U_{\varepsilon}(0) = V_{\varepsilon}(0)$. 
Then the low-rank Kalman-Bucy filter is defined as
\begin{align*}
d \tilde{x}_{\varepsilon} (t) =& A(t) \tilde{x}_{\varepsilon}(t) dt 
\nonumber \\ & + U_{\varepsilon}(t) P_{r,\varepsilon}(t) C_{U_{\varepsilon}}(t) ^{\top} 
\nonumber \\ & \quad \times (dy(t) - C \tilde{x}_{\varepsilon}(t)dt),
\\
\frac{d}{dt} P_{r,\varepsilon} (t) =& A_{U_{\varepsilon}}(t) P_{r,\varepsilon} + P_{r,\varepsilon} (t) A_{U_{\varepsilon}}(t)^{\top} 
\nonumber \\ & + G_{U_{\varepsilon}}(t) G_{U_{\varepsilon}}(t) ^{\top} \nonumber \\ &-  P_{r,\varepsilon} (t) C_{U_{\varepsilon}}(t)^{\top} C_{U_{\varepsilon}}(t) P_{r,\varepsilon} (t),
\end{align*}
where $\tilde{x}(t) \in {\mathbb{R}}^{n}$ is the estimated state,  
$P_{r,\varepsilon} (t) \in {\mathbb{R}}^{r\times r}$, $A_{U_{\varepsilon}}(t) := U_{\varepsilon}(t)^{\top} A(t) U_{\varepsilon}(t)$, \\ $G_{U_{\varepsilon}}(t) := (V_{\varepsilon}(t)^{\top}U_{\varepsilon}(t))^{-1} V_{\varepsilon}(t)^{\top} G$, $C_{U_{\varepsilon}}(t) := C U_{\varepsilon}(t)$. 
From the result of our follow-up paper \cite{TsuzukiOhki2025global}, the definition of $G_{U_{\varepsilon}}(t)$ is different from the original proposal of the low-rank Kalman-Bucy filter \cite{yamada2020lowrank}.  
The covariance matrix $P_{\varepsilon}(t) := {\mathbb{E}}[ (x(t) - \tilde{x}_{\varepsilon}(t)) (x(t) - \tilde{x}_{\varepsilon}(t))^{\top} ] $ is obtained from the following differential Lyapunov equation \cite{yamada2020lowrank}. 
\begin{align}
&\frac{d}{dt}P_{\varepsilon}(t) 
\nonumber \\ 
= & (A(t) - U_{\varepsilon}(t) P_{r,\varepsilon} (t)C_{U_{\varepsilon}}(t)^{\top} C ) P_{\varepsilon}(t) 
\nonumber \\
&+
P_{\varepsilon}(t) (A(t) - U_{\varepsilon}(t) P_{r,\varepsilon}(t) C_{U_{\varepsilon}}(t)^{\top} C )
\nonumber \\ &
 + U_{\varepsilon}(t) P_{r,\varepsilon}(t) C_{U_{\varepsilon}}(t)^{\top} C_{U_{\varepsilon}}(t) P_{r,\varepsilon}(t) U_{\varepsilon}(t)^{\top}
 \nonumber \\ &
 + GG^{\top}
\label{eq:covariance_LRKF}
\end{align}
We demonstrate the performance of the low-rank Kalman-Bucy filter with $\varepsilon =1, 0.1, 0.01$. 
For the comparison with the best filter (Kalman-Bucy filter), we plotted the ratio between the trace of $P_{\varepsilon}(t)$ and $P(t)$ for each $\varepsilon$. 
By Euler method with the time step size $dt = 10^{-4}$, Figure \ref{fig_LRKBF_covariance} shows the result for $n=10$. 
Figure \ref{fig_example_tracking_semilog} demonstrates the tracking error of the time-varying dominant modes, which is represented $R(t)\Psi_{r} \in {\mathbb{C}}^{n\times r}$, where $\Psi _r$ the first $r$ eigenvectors of $\Psi (A)$. 
From the result, if we take small $\varepsilon$, then the estimation  accuracy remains near the Kalman-Bucy filter case. 
Since the spectral gap of $A$ is small, small time step size is required for the stable numerical simulation.  
Because this requires high computational power, this should be resolved in the future. 

From both the analysis and simulations, $\varepsilon$ should be selected by balancing tracking accuracy and numerical stiffness: smaller $\varepsilon$ improves subspace tracking, whereas excessively small $\varepsilon$ may increase computational burden and numerical sensitivity.

   	\begin{figure}[!htbp]
    	\centering
    	\includegraphics[keepaspectratio,width=\linewidth]{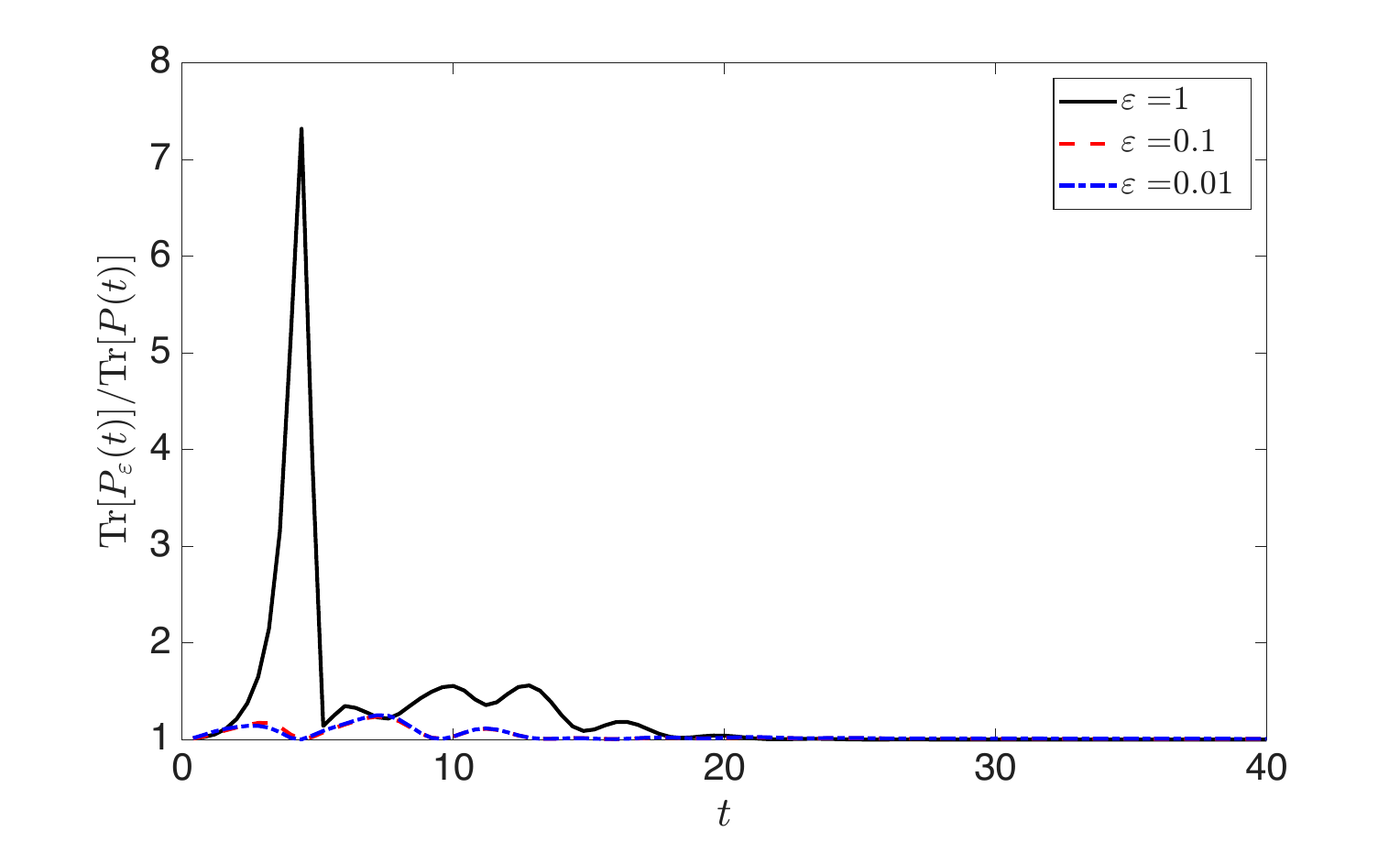}
    	\caption{Plot of ${\mathrm{Tr}}[P_{\varepsilon}(t)] / {\mathrm{Tr}}[P(t)] $ for each time $t$ under different convergence parameter $ \varepsilon = 1, 0.1, 0.01$. }
	\label{fig_LRKBF_covariance}
	\end{figure}

   	\begin{figure}[!htbp]
    	\centering
    	\includegraphics[keepaspectratio,width=\linewidth]{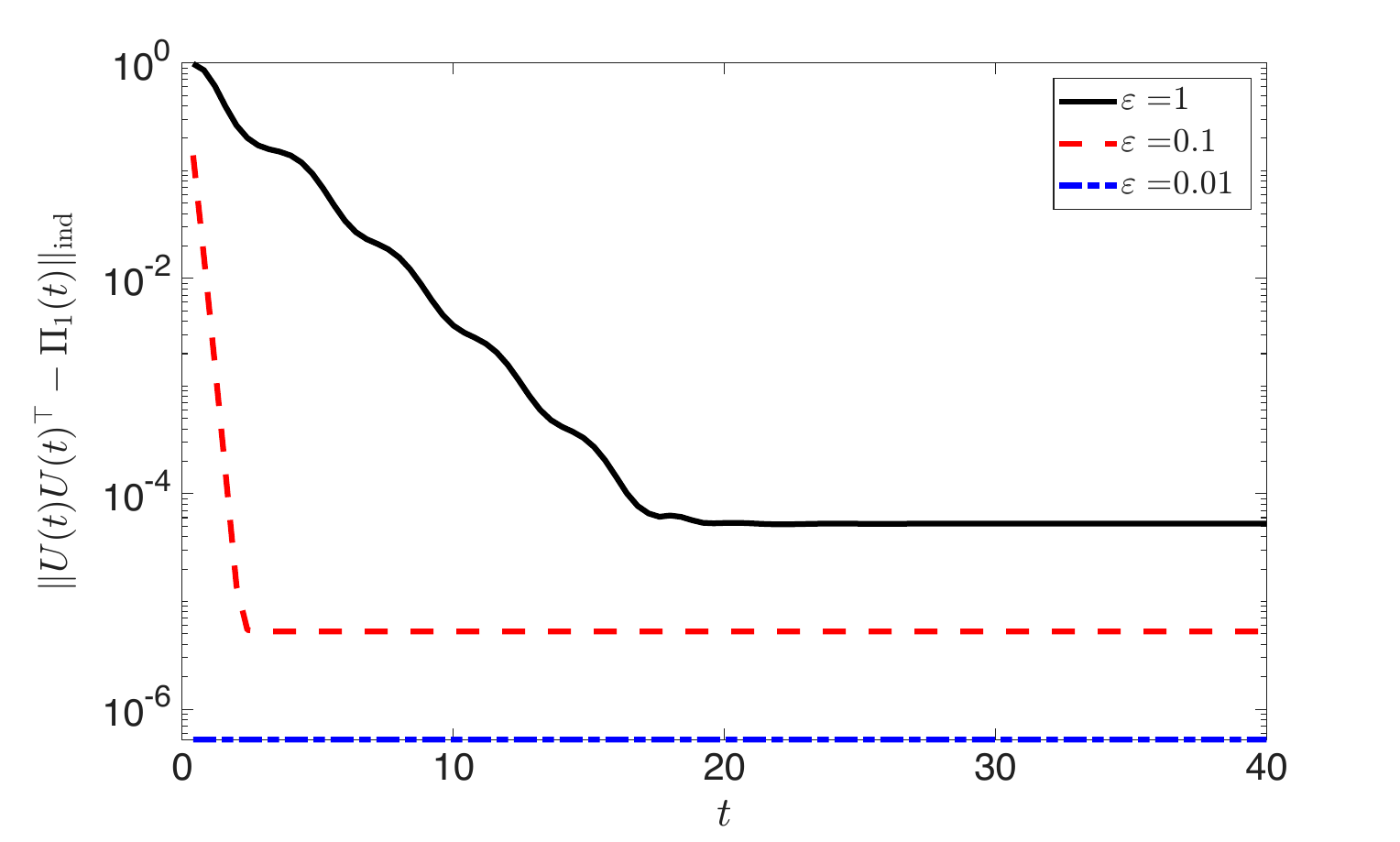}
    	\caption{Plot of $\| U(t) U(t)^{\top} - R(t) \Psi _{r} (\Psi _{r}^{\top} \Psi_r)^{-1} \Psi _{r}^{\top}R(t)^{\top} \| _{\rm ind}$ for each $t$. }
	\label{fig_example_tracking_semilog}
	\end{figure}

\section{Conclusion}

In this paper, we investigate the tracking performance of the Oja flow for time-varying matrices. 
Under certain assumptions, a tracking error was established theoretically. 
A numerical example shows that the correctness of the theoretical results and challenges in practice. 
We will analyze the estimation accuracy of the low-rank Kalman filter for LTV systems in the future work. 

\section*{Acknowledgement}

This work was supported by JSPS KAKENHI Grant Number JP26K07550. 


\appendix
\section{Proof of Proposition \ref{prop:time-varying_positive_part_convergence}}
\label{proof:time-varying_positive_part_convergence}

We first introduce the following lemma that generalizes the results in \cite{sasagawa1982finite}.

    \begin{lemma}[{\cite[Theorems 13 and 14]{kilicaslan2012existence}}]
    \label{lemma:kilicaslan2012existence}
    Consider the following Riccati differential equation;
	\begin{align*}
	\frac{d}{dt}P(t) =& A(t)P(t) + P(t) A(t)^{\top} 
	\\ & - P(t) Q(t) P(t) + R(t)
	, 
	\end{align*}
	where $A(t), Q(t)=Q(t)^{\top}, R(t) = R(t)^{\top} \in {\mathbb{R}}^{n\times n}$, $Q(t)$, $R(t) \geq 0$, and $P(t_{0})=P_{0} \geq 0$. 
	Let $\bar{P}(t) \in {\mathbb{R}}^{n\times n}$ be one of the equilibrium points of the Riccati differential equation corresponding to the values of $A(t)$, $Q(t)$, and $Q(t)$ evaluated at time $t$.  
	Then, $P(t) = Y(t)X(t)^{-1}$ for all $t \in [t_{0},t_{\max})$, where 
	\begin{align*}
	 X(t) 
	=& \hat{\Phi} _{\hat{A}} (t,t_{0}) 
	\\ & 
	+ \int _{t_{0}}^{t} \hat{\Phi}_{\hat{A}}(t,s) R(t) \check{\Phi}_{\check{A}} (s,t_{0}) ds 
	\\ & \quad \quad \times 
	(P_{0} - \bar{P}(t_{0}))  
	,\\
	 Y(t) 
	 =& \bar{P}(t_{\max}) \hat{\Phi}_{\hat{A}} (t,t_{0})  \\ &
	+ 
	\bar{P}(t_{\max}) \int _{t_{0}}^{t} \hat{\Phi}_{\hat{A}}(t,s) R(t) \check{\Phi}_{\check{A}} (s,0) ds 
	\\ & \quad \quad \times (P_{0} - \bar{P}(t_{0}))  
	\\ &  
	+ \hat{\Phi}_{\hat{A}}(t,t_{0})( P_{0} - \bar{P}(t_{0})  )
	,
	\end{align*} 
	$\hat{\Phi}_{\hat{A}}$ and $\check{\Phi}_{\check{A}}$ are transition matrices with respect to $\hat{A}(t) := -A(t)^{\top} + R(t) \bar{P}(t) $ and $\check{A}(t) := A(t) - \bar{P}(t)R(t) $, respectively, 
	$X(t_{0})=I_{n}$, $Y(t_{0}) = P_{0}$, 
	and $t_{\max} := \inf \{ t \geq t_{0} \ | \ {\mathrm{det}} (X(t)) =0 \}$. 
    \end{lemma}

    From this lemma, we have a generalization of Lemmas 7 and 8 in \cite{TsuzukiOhki2025global}.

    \begin{lemma} \label{lem:well_definedness_time_varying_Oja_flow}
    For a given matrix-valued bounded function $A: {\mathbb{R}} \to {\mathbb{R}}^{n\times n}$ and arbitrarily positive constant $c>0$, take $a>0$ such that $A_{\rm sym}(t) + aI_{n} >cI_{n}$ for all $t\in \mathbb{R}$. 
    If $U(t_{0}) = U_{0} \in {\mathbb{R}}^{n\times r}$ is a full rank matrix, then $U(t)$ remains full rank for all $t\geq t_{0}$ and its singular values converges to $1$. 
    Furthermore, $\lambda _{r}( U(t)^{\top} U(t)) \geq \alpha$, $\alpha := \min \{ 1 , \lambda _{r} (U_{0}^{\top}U_{0}) \}$ for all $t\geq t_{0}$. 
    \end{lemma}
    
    \begin{proof}
    Let $B(t) := A (t) + aI_{n}$. 
    For simplicity, we only consider $\varepsilon =1$. 
    Define $P(t) := U(t) U(t)^{\top}$, and then its time evolution equation is
    \begin{align}
    \frac{d}{dt} P(t) = & B (t) P(t) + P(t) B(t)^{\top} \nonumber
    \\ & 
    - 2 P(t) B_{\rm sym}(t) P(t) 
    \label{eq:time-varying_riccati_oja_flow}
    \end{align}
    with $P(t_{0}) = P_{0} = U_{0} U_{0} ^{\top}$. 
    Since $ \bar{P}= 0_{n,n}$ is an equilibrium of \eqref{eq:time-varying_riccati_oja_flow} at all $t \in \mathbb{R}$, from Lemma \ref{lemma:kilicaslan2012existence}, the solution of \eqref{eq:time-varying_riccati_oja_flow} is as following: 
    \begin{align*}
    P(t) = \hat{\Phi }_{B}(t,t_{0}) P_{0} 
    (I_{n} + G(t,t_{0}) P_{0}) ^{-1}
    \hat{\Phi}_{B^{\top}}(t,t_{0})  ,
    \end{align*}
    where 
    \begin{align*}
    G(t,t_{0}) :=   2 \int _{t_{0}}^{t} \check{\Phi}_{B^{\top}} (s,t_{0})  B_{\rm sym}(s) \check{\Phi}_{B}(s, t_{0}) ds 
    .
    \end{align*}
    Notice that since $B(\bullet ) = A(\bullet ) + aI_{n}$ is bounded, the transition matrix $\hat{\Phi} _{B} (t,t_{0})$ always has its inverse. 
    The similar arguments in the proof of Lemma 7 in \cite{TsuzukiOhki2025global} show that $ ( I_{n} + G(t,t_{0}) P_{0} )$ is invertible for all $t \geq t_{0}$. 
    Therefore, Equation \eqref{eq:time-varying_Oja_flow} preserves the rank condition of $U(t)$ for all $t \geq t_{0}$. 
    As in the proof of Lemma 8 in \cite{TsuzukiOhki2025global}, 
    \begin{align*}
    G(t,t_{0}) = \int_{t_{0}}^{t} \frac{d}{ds}H(s,t_{0}) ds = H(t,t_{0}) - I_{n}
    \end{align*}
    holds, where $H(t,t_{0}) = \check{\Phi}_{B^{\top}} (t,t_{0}) \check{\Phi}_{B} (t,t_{0})$.  
    From the assumption, $B_{\rm sym}(t) > cI_{n}$ for all $t\geq t_{0}$ and therefore, $\frac{d}{dt} H(t,t_{0}) > c H(t,t_{0}) > 0$ for all $t\geq 0$. 
    Since $H(t,t_{0})$ diverges exponentially, $\check{\Phi}_{B} (t,t_{0})$ also diverges. 
    The same arguments of the proof of Lemma 8 in \cite{TsuzukiOhki2025global} conclude that $U(t)$ is of full rank for all $t\geq t_0$, its singular values converge to $1$, and the minimal eigenvalue of $U(t)^{\top}U(t)$ is bounded below by $\alpha = \min \{ 1, \lambda _{r} (U_{0}^{\top}U_{0}) \}$. 
    \end{proof}

	From the above results, the same argument of the proof of Theorem 9 in \cite{TsuzukiOhki2025global} concludes the statement of Prop. \ref{prop:time-varying_positive_part_convergence}.

\end{document}